\documentclass{amsart}
\usepackage{amsmath,amssymb,multirow,color,placeins,graphicx,realboxes,afterpage,tikz,cite,verbatim}
\usepackage{etoolbox}

\numberwithin{equation}{section}

\newtheorem{prop}[equation]{Proposition}

\newtheorem{cor}[equation]{Corollary}

\theoremstyle{definition}
\newtheorem*{rmk}{Remark}

\newtheorem{defn}[equation]{Definition}

\newcommand{\F}{\mathbb{F}}
\newcommand{\bP}{\mathbb{P}}

\newcommand{\Z}{\mathbb{Z}}

\DeclareMathOperator{\charp}{char}

\renewcommand{\bar}[1]{#1\llap{$\overline{\phantom{\rm#1}}$}}

\usepackage[colorlinks,pagebackref,pdftex, bookmarks=false]{hyperref}

\begin{document}

\title{Tangent-Chebyshev rational maps and R\'edei functions}

\author{Zhiguo Ding}
\address{
  Hunan Institute of Traffic Engineering,
  Hengyang, Hunan 421001 China
}
\email{ding8191@qq.com}

\author{Michael E. Zieve}
\address{
  Department of Mathematics,
  University of Michigan,
  530 Church Street,
  Ann Arbor, MI 48109-1043 USA
}
\email{zieve@umich.edu}
\urladdr{http://www.math.lsa.umich.edu/$\sim$zieve/}

\date{\today}

\begin{abstract}
Recently Lima and Campello de Souza introduced a new class of rational functions over odd-order finite fields, and explained their potential usefulness in cryptography.  We show that these new functions are conjugate to the classical family of R\'edei rational functions, so that the properties of the new functions follow from properties of R\'edei functions.  We also prove new properties of these functions, and introduce analogous functions in characteristic $2$, while also introducing a new version of trigonometry over finite fields of even order,  which is of independent interest.
\end{abstract}

\thanks{
The second author thanks the National Science Foundation for support under grant DMS-1601844.}

\maketitle

%#######################################################################
%#######################################################################

\section{Introduction}

Recently Lima and Campello de Souza \cite{LS} introduced a family of rational functions over odd-characteristic finite fields, which they called tangent-Chebyshev maps.
They showed that these functions have properties  suitable for use in cryptography.  In this note we prove a relationship between these tangent-Chebyshev maps and a classical family of rational functions introduced by R\'edei \cite{R}, and prove several new properties of tangent-Chebyshev maps.  This allows one to deduce cryptographic information about tangent-Chebyshev maps from the known cryptographic results involving R\'edei functions \cite{BM,LM,LM2,MW,More,MN,MS,NRup1,NRup3}.
We also introduce characteristic $2$ analogues of these maps, based on a new characteristic $2$ version of finite field trigonometry which is of independent interest.

%#######################################################################
%#######################################################################

\section{R\'edei functions}

We begin with a quick development of the theory of R\'edei functions.  Let $q$ be a prime power, and let $h(x)\in\F_q[x]$ be an irreducible polynomial of the form
\[
\begin{array}{c@{\mathrel{}{}}l@{}}
x^2-\alpha &\qquad\text{ if $q$ is odd} \\[\smallskipamount]
x^2+x+\alpha&\qquad\text{ if $q$ is even}.\smallskip
\end{array}
\]
\noindent
It is known that for each $q$ there are exactly $\lfloor q/2\rfloor$ such
irreducible polynomials $h(x)$.
%
\begin{comment}
{\color{blue} (For any $n>1$, since each such monic irreducible $h(x) \in \F_{q^2}[x]$ corresponds exactly to 
one R\'edei function $R_n$ defined in this paper, it may help some readers to count the number $N_q$ of all such 
polynomials $h(x)$ for any $q$, which equals the number of degree-$n$ R\'edei functions $R_n$ over $\F_q$ for 
any $n>1$. Indeed, it is obvious that $N_q$ equals $\lfloor q/2 \rfloor$. --- But this may not be very interesting, 
since all degree-$n$ R\'edei functions are in the same equivalence class up to conjugation with a degree-one 
rational function over $\F_q$. I mention it just in order to make reviewers happier if it could.)}{\color{red}We also need to exclude $n$ being a power of $\charp(\F_q)$.  I tried to do this in a remark at the end of this section, but this exclusion makes things inelegant.}
\end{comment}
%
Let $\beta\in\F_{q^2}$ be a root of $h(x)$, so that $\bar\beta:=\beta^q$ is the other root of $h(x)$, namely
\[
\begin{array}{c@{\mathrel{}{}}l@{}}
\bar\beta=-\beta &\qquad\text{ if $q$ is odd} \\[\smallskipamount]
\bar\beta=\beta+1&\qquad\text{ if $q$ is even}.\smallskip
\end{array}
\]

\begin{defn}
For any positive integer $n$, the degree-$n$ \emph{R\'edei function} over\/ $\F_q$ with parameter $\alpha$ is
\[
R_n(x,\alpha):=\rho^{-1} \circ x^n \circ \rho
\]
where $\rho(x):=(x-\bar\beta)/(x-\beta)$ and $\rho^{-1}(x):=(\beta x - \bar\beta)/(x-1)$ are degree-one rational functions in $\F_{q^2}(x)$ such that $\rho^{-1}\circ\rho=x$.
\end{defn}

Explicitly,
\[
R_n(x,\alpha)=\frac{\beta (x-\bar\beta)^n - \bar\beta (x-\beta)^n}{(x-\bar\beta)^n-(x-\beta)^n}.
\]
Thus $R_n(x,\alpha)$ is unchanged if we interchange $\beta$ and $\bar\beta$, so that $R_n(x,\alpha)$ does not depend on the choice of $\beta$.  Moreover, since the $q$-th power map interchanges $\beta$ and $\bar\beta$, it fixes each coefficient of $R_n(x,\alpha)$, so that $R_n(x,\alpha)\in\F_q(x)$.

We compute the coefficients of $R_n(x,\alpha)$ via the binomial theorem.  For odd $q$ this yields
\[
R_n(x,\alpha)=\frac{\sum_{i=0}^{\lfloor n/2\rfloor}\binom{n}{2i} \alpha^i x^{n-2i}}
{\sum_{i=0}^{\lfloor (n-1)/2\rfloor} \binom{n}{2i+1} \alpha^i x^{n-2i-1}},
\]
while for even $q$ the function $R_n(x,\alpha)$ equals
\[
\frac{\sum_{i=0}^n \Bigl((\beta+1)\beta^i + \beta (\beta+1)^i\Bigr)\binom{n}{i} x^{n-i}}
{\sum_{i=0}^n \Bigl(\beta^i+(\beta+1)^i\Bigr)\binom{n}{i}x^{n-i}}.
\]

The function $\rho(x)$ induces a bijection from $\bP^1(\F_q):=\F_q\cup\{\infty\}$ to the set $\mu_{q+1}$ of $(q+1)$-th roots of unity in $\F_{q^2}$ \cite{ZR}.  Thus all properties of the function induced by $R_n(x,\alpha)$ on $\bP^1(\F_q)$ correspond to properties of the function induced by $X^n$ on $\mu_{q+1}$.  This may be formally expressed in terms of functional graphs:

\begin{defn} For any set $S$ and any function $f\colon S\to S$, the \emph{functional graph} of $f$ is the directed graph with vertices labelled by $S$, and with an edge from $s$ to $t$ if and only if $f(s)=t$.
\end{defn}

\begin{prop}
The functional graph of $R_n(x,\alpha)$ on\/ $\bP^1(\F_q)$ is isomorphic to the functional graph of $x^n$ on $\mu_{q+1}$.
\end{prop}

\begin{cor}
The function $R_n(x,\alpha)$ permutes\/ $\bP^1(\F_q)$ if and only if $\gcd(n,q+1)=1$.
\end{cor}

\begin{proof}
We see that $R_n(x,\alpha)$ permutes $\bP^1(\F_q)$ if and only if $x^n$ permutes $\mu_{q+1}$, or equivalently $\gcd(n,q+1)=1$.
\end{proof}

\begin{prop}
The functions $R_n(x,\alpha)$ and $R_m(x,\alpha)$ commute under composition.
\end{prop}

\begin{proof}
This follows from the fact that $x^n$ and $x^m$ commute under composition.
\end{proof}

We now prove the following identity which seems to be new:

\begin{prop}\label{fe}
For any $m,n>0$, the function $R_{m+n}(x,\alpha)$ equals
\[
\frac{R_m(x,\alpha)R_n(x,\alpha)+\alpha}
{R_m(x,\alpha)+R_n(x,\alpha)-(\beta+\bar\beta)}.
\]
\end{prop}

%
\begin{comment}
% Magma verification:

_<b,bb,m,n>:=PolynomialRing(Rationals(),4);
u:=(m-bb)/(m-b) * (n-bb)/(n-b);
(b*u-bb)/(u-1) eq (m*n-b*bb)/(m+n-(b+bb));
// and  -b*bb = alpha

\end{comment}
%

\begin{rmk}
In the above expression, we note that $\beta+\bar\beta$ is $0$ for odd $q$ and $1$ for even $q$.
\end{rmk}

\begin{proof}
We compute
\begin{align*}
\rho\,\circ\,& R_{m+n}(x,\alpha) = \rho(x)^{m+n} \\
&=\rho(x)^m\cdot\rho(x)^n \\
&=\rho(R_m(x,\alpha))\cdot\rho(R_n(x,\alpha)) \\
&=\frac{R_m(x,\alpha)-\bar\beta}{R_m(x,\alpha)-\beta}\cdot\frac{R_n(x,\alpha)-\bar\beta}{R_n(x,\alpha)-\beta},
\end{align*}
and the result follows by composing on the left with $\rho^{-1}(x)$.
\end{proof}

\begin{rmk}
For odd $q$, the case $m=n$ of Proposition~\ref{fe} is essentially the case $y=x$ of the identity of bivariate rational functions in \cite[\S 4]{C}, after one observes that $R_{2n}(x,\alpha)=R_n(x,\alpha)\circ R_2(x,\alpha)$.  It would be interesting to explore bivariate versions of other special cases of Proposition~\ref{fe}.
\end{rmk}

\begin{rmk}
One can show that if $n$ is not a power of the characteristic of $\F_q$ then distinct values $\alpha$ yield distinct rational functions $R_n(x,\alpha)$.  Thus, for any such $q$ and $n$, the number of degree-$n$ R\'edei functions over $\F_q$ is the number of possibilities for $\alpha$, namely $\lfloor q/2\rfloor$.
%
% Is there a really quick and elementary proof suitable for this audience?  If gcd(n,q)=1 then the critical points of R_n are rho^{-1}({0,infinity})={beta,bar beta}.  But then one has to do something to deduce uniqueness when gcd(n,q)>1.  Better: if  r^{-1} o x^n o r = s^{-1} o x^n o s  then both sides have the same set of points with a unique preimage in P^1(\bar\F_q).  If n isn't a power of char(F_q) then this says r^{-1}({0,infinity})=s^{-1}({0,infinity}), i.e., the beta,bar beta-values for r equal those for s, so r,s have the same alpha-value.
\end{rmk}

The arguments in this section immediately extend to the case that $h(x)$ is an arbitrary irreducible degree-$2$ polynomial in $\F_q[x]$.  We required $h(x)$ to be $x^2-\alpha$ or $x^2+x+\alpha$ since these are the polynomials occurring in the original papers on the topic \cite{N2,R}, and also these yield simpler expressions for the coefficients of $R_n(x,\alpha)$.

%#######################################################################
%#######################################################################

\section{Tangent-Chebyshev functions}

The tangent-Chebyshev rational functions $C_n(x,\alpha)$ were defined in \cite{LS} for any odd prime power $q$ via a somewhat involved procedure.  This led to the explicit expression \cite[(9)]{LS}, which we use as our definition:

\begin{defn}
Let $n$ be a positive integer, let $q$ be an odd prime power, and let $\alpha$ be a nonsquare in $\F_q$.  Then the degree-$n$ tangent-Chebyshev rational function over $\F_q$ with parameter $\alpha$ is
\[
C_n(x,\alpha):=\frac{\sum_{i=0}^{\lfloor (n-1)/2\rfloor} \binom{n}{2i+1}\alpha^i x^{2i+1}}
{\sum_{i=0}^{\lfloor n/2\rfloor}\binom{n}{2i} \alpha^i x^{2i}}.
\]
\end{defn}

Comparing this definition to the explicit expression for $R_n(x,\alpha)$ yields the following new and fundamental equality:
\begin{equation}\label{CR}
C_n(x,\alpha)=\frac1x\circ R_n(x,\alpha)\circ\frac1x.
\end{equation}

Conversely, it is natural to take \eqref{CR} to be the definition of $C_n(x,\alpha)$ in case $q$ is an even prime power and $\alpha\in\F_q$ equals $\beta^2+\beta$ with $\beta\in\F_{q^2}\setminus\F_q$.
By the results of the previous section, in this case $C_n(x,\alpha)$ equals
\[
\frac
{\sum_{i=0}^n \Bigl(\beta^i+(\beta+1)^i\Bigr)\binom{n}{i}x^{i}}
{\sum_{i=0}^n \Bigl((\beta+1)\beta^i + \beta (\beta+1)^i\Bigr)\binom{n}{i} x^{i}}.
\]

In the rest of this section, $q$ is any prime power (odd or even), and $\alpha$, $\beta$, $\bar\beta$ are as at the start of the previous section.
We now present consequences of \eqref{CR}.

\begin{cor}\label{7}
We have
\[
C_n(x,\alpha)=\eta^{-1}\circ x^n\circ\eta
\]
where $\eta(x):=\rho(1/x)=(\bar\beta x-1)/(\beta x-1)$ and $\eta^{-1}(x)=1/\rho^{-1}(x)=(x-1)/(\beta x-\bar\beta)$.  Here $\eta^{-1}\circ\eta=x$, and $\eta$ maps\/ $\bP^1(\F_q)$ bijectively onto $\mu_{q+1}$.
\end{cor}

\begin{cor}\label{LS7}
The rational functions $C_m(x,\alpha)$ and $C_n(x,\alpha)$ commute under composition.
\end{cor}

\begin{proof}
This follows from the previous corollary since $x^m$ and $x^n$ commute.
\end{proof}

\begin{cor}\label{9}
The function $C_n(x,\alpha)$ permutes\/ $\bP^1(\F_q)$ if and only if $\gcd(n,q+1)=1$.
\end{cor}

\begin{proof}
Since $\eta$ maps $\bP^1(\F_q)$ bijectively onto $\mu_{q+1}$, we see that $C_n(x,\alpha)$ permutes $\bP^1(\F_q)$ if and only if $x^n$ permutes $\mu_{q+1}$, or equivalently $\gcd(n,q+1)=1$.
\end{proof}

\begin{cor}\label{LS10}
For odd $q$, the function $C_n(x,\alpha)$ permutes\/ $\F_q$ if and only if $\gcd(n,q+1)=1$.
\end{cor}

\begin{proof}
If $n$ is even then $\gcd(n,q+1)$ is even and 
$C_n(x,\alpha)$ does not permute $\F_q$ since
$C_n(0,\alpha)=\infty$.
Now suppose $n$ is odd.  Then $C_n(x,\alpha)$ 
fixes $\infty$, and hence permutes $\F_q$ if and only if it permutes $\bP^1(\F_q)$, which says $\gcd(n,q+1)$ by the previous result.
\end{proof}

\begin{comment}
{\color{blue} (For odd $q$, in the case $n$ is odd, $C_n(x,\alpha)$ fixes $\infty$ since $R_n$ fixes $0$; 
in the case $n$ is even, $C_n(x,\alpha)$ maps $\infty$ to $0$ since $R_n$ maps $0$ to $\infty$, so $C_n(x,\alpha)$ 
does not fix $\infty$, but Proposition~\ref{LS10} remains true since in this case we have both $C_n$ does 
not permute $\F_q$ and $\gcd(n,q+1)\ne 1$.}
{\color{red}Yes, I should mention even $n$.  But I don't see why $\infty\mapsto 0$ prevents bijectivity on $\F_q$.}{\color{blue} 
--- Zhiguo: it does not prevent bijectivity on $\F_q$, but it prevents $C_n$ permutes $\F_q$.}{\color{red}  WHY???  Why can't $C_n$ permute $\F_q$ while also mapping $\infty\to 0$ (so it doesn't permute $\bP^1(\F_q)$)?  I think we need to justify this.}
{\color{blue}Moreover, for any $q$ and $n\ge 1$, it is clear that $C_n$ permutes 
$\F_q$ implies $C_n$ permutes $\bP^1(\F_q)$, which is equivalent to $\gcd(n,q+1)=1$ by Proposition~\ref{9}. 
If in addition $q$ is odd then the converse holds since $C_n$ fixes $0$. If $q$ is even then the converse 
is not necessary true, for example when $n=2$. But this is not important.)
}
\end{comment}

\begin{cor}\label{8}
The functional graph of $C_n(x,\alpha)$ on\/ $\bP^1(\F_q)$ is isomorphic to the functional graph of $x^n$ on $\mu_{q+1}$.
\end{cor}

\begin{cor}\label{LS11}
The function $C_n(x,\alpha)$ induces an involution on\/ $\bP^1(\F_q)$ if and only if $n^2\equiv 1\pmod{q+1}$.
\end{cor}

\begin{proof}
The function $C_n(x,\alpha)$ induces an involution on $\bP^1(\F_q)$ if and only if $x^n$ induces an involution on $\mu_{q+1}$, or equivalently $x^n\circ x^n=x^{n^2}$ induces the identity map on $\mu_{q+1}$.  This says $x^{n^2-1}$ maps each element of $\mu_{q+1}$ to $1$, which says $(q+1)\mid (n^2-1)$ since $\mu_{q+1}$ is a cyclic group of order $q+1$.
\end{proof}

The above corollaries include generalizations of the main results from \cite{LS}.  Specifically, \cite[Prop.~7]{LS} is the odd $q$ case of Corollary~\ref{LS7}; \cite[Prop.~11]{LS} is the odd $q$ case of Corollary~\ref{LS11}; \cite[Prop.~10]{LS} is Corollary~\ref{LS10}, which follows from the odd $q$ case of Corollary~\ref{9};
and \cite[Prop.~5 and 6]{LS} describe the zeroes and poles of $C_n(x)$ for odd $q$, which is immediate for all $q$ from Corollary~\ref{7}.
Corollaries~\ref{7} and \ref{8} are new, and explain much about the functions $C_n(x,\alpha)$.

Next, Proposition~\ref{fe} yields the following new identity, which seems useful for calculating values of $C_n(x,\alpha)$ and for analyzing cryptosystems built from those values:

\begin{prop}\label{Cadd}
For any $m,n>0$, we have
\[
C_{m+n}(x,\alpha) = H(C_m(x,\alpha),\, C_n(x,\alpha))
\]
where
\[
H(x,y):=\frac{x+y-(\beta+\bar\beta)xy}{1+\alpha xy}.
\]
\begin{comment}
the function $C_{m+n}(x,\alpha)$ equals
\[
\frac{C_m(x,\alpha)+C_n(x,\alpha)-(\beta+\bar\beta)C_m(x,\alpha)C_n(x,\alpha)}{1+\alpha C_m(x,\alpha)C_n(x,\alpha)}
.
\]
\end{comment}
\end{prop}

\begin{rmk}
If $q$ is odd then $H(x,y)=(x+y)/(1+\alpha xy)$, which intriguingly coincides with Einstein's velocity-addition formula in special relativity when $\alpha=-1/c^2$ with $c$ being the speed of light in a vacuum \cite[\S 5]{E}.
% alternately: \cite[(19.22)]{A}
\end{rmk}

%#######################################################################
%#######################################################################

\section{Trigonometry in characteristic $2$}

Analogues of $\cos x$ and $\sin x$ in $\F_q$ were introduced for odd $q$ in \cite{SOKP}, where they were used to define a Hartley number-theoretic transform.  Analogues of the tangent and inverse tangent function for odd $q$ were introduced in \cite{LS} and used in the definition of $C_n(x,\alpha)$ for odd $q$.
An analogue of $\cos x$ was introduced in \cite{LBS} for even $q$.
It appears that the literature does not contain an analogue of $\sin x$ or $\tan x$ over finite fields of even order.  We introduce such an analogue now, by combining our definition of $C_n(x,\alpha)$ for even $q$ from the previous section with the analogy with classical trigonometric functions that motivated the definition of $C_n(x,\alpha)$ in \cite{LS}.

For the rest of this section, $q$ is an even prime power, and we fix some $\zeta\in\F_{q^2}^*$ and some $\beta\in\F_{q^2}\setminus\F_q$ such that $\alpha:=\beta^2+\beta$ is in\/ $\F_q$.  Write $\bar\beta:=\beta^q$, so that as noted previously we have $\bar\beta=\beta+1$.  For any integer $k$ we define
\begin{align*}
\sin_{\zeta}(k)&:=\zeta^k+\zeta^{-k}
\\
\cos_{\zeta,\beta}(k)&:=\beta\zeta^k+\bar\beta\zeta^{-k}
\\
\tan_{\zeta,\beta}(k)&:=\frac{\sin_{\zeta}(k)}{\cos_{\zeta,\beta}(k)}.
\end{align*}
We note that what we call $\sin_{\zeta}(k)$ is called $\cos_\zeta(k)$ in \cite{LBS}; but we will show that the properties of this function resemble those of $\sin k$ rather than $\cos k$.

\begin{comment}
K<b,u,v>:=FunctionField(GF(2),3);
a:=b^2+b;
bb:=b+1;
sx:=u+1/u;
cx:=b*u+bb/u;
tx:=sx/cx;
sy:=v+1/v;
cy:=b*v+bb/v;
ty:=sy/cy;
ssum:=u*v+1/(u*v);
csum:=b*u*v+bb/(u*v);
tsum:=ssum/csum;

tsum eq (tx+ty+tx*ty)/(1+a*tx*ty);

csum eq cx*cy + a*sx*sy;

ssum eq sx*cy + sy*cx + sx*sy;

cx^2 + sx*cx + a*sx^2 eq 1;

u eq cx + bb*sx;

Maybe Neubauer's definition of char 2 Redei functions is somehow implicitly already conjugated by 1/x when compared to the odd char version, so that our char 2 C_n comes from cotangent rather than tangent???

K<b,u,v>:=FunctionField(Rationals(),3);
a:=b^2;
bb:=-b;
cx:=(u+1/u)/2;
sx:=(u-1/u)/(2*a);
tx:=sx/cx;
cy:=(v+1/v)/2;
sy:=(v-1/v)/(2*a);
ty:=sy/cy;
csum:=(u*v+1/(u*v))/2;
ssum:=(u*v-1/(u*v))/(2*a);
tsum:=ssum/csum;

ssum eq sx*cy+sy*cx;
csum eq cx*cy+a^2*sx*sy;
u eq cx+a*sx;
cx^2-a^2*sx^2 eq 1;

// the next identity is new for odd q

tsum eq (tx+ty)/(1+a^2*tx*ty);

\end{comment}

The crucial property of $\tan_{\zeta,\beta}(x)$ is as follows:

\begin{prop}\label{foo}
For $k\in\Z$ we have
\[
C_n(\tan_{\zeta,\beta}(k),\alpha) = \tan_{\zeta,\beta}(nk).
\]
\end{prop}

This result is the analogue for even $q$ of the equation used in \cite{LS} for odd $q$ to define $C_n(x,\alpha)$.

\begin{proof}
Writing $x:=\tan_{\zeta,\beta}(k)$ and $u:=\zeta^k$, we have
\[
x=\frac{u+u^{-1}}{\beta u+\bar\beta u^{-1}}=\eta^{-1}(u^2)
\]
and likewise
\[
\tan_{\zeta,\beta}(nk)=\eta^{-1}(u^{2n}).
\]
Thus
\begin{align*}
\eta(C_n(x,\alpha))&=\eta(x)^n \\
&= u^{2n} \\
&= \eta(\tan_{\zeta,\beta}(nk)),
\end{align*}
as desired.
\end{proof}

Our characteristic $2$ trigonometric functions satisfy the following easily verified properties for any integers $k,\ell$; here we suppress the subscripts for ease of typography:

\begin{align*}
\zeta^k &= \cos(k)+\bar\beta\sin(k) \\
1 &=\cos(k)^2 + \sin(k)\cos(k)+\alpha\sin(k)^2 \\
\sin(k+\ell)&=\sin(k)\cos(\ell)+\sin(\ell)\cos(k)+\sin(k)\sin(\ell) \\
\cos(k+\ell) &=\cos(k)\cos(\ell)+\alpha\sin(k)\sin(\ell) \\
\tan(k+\ell) &= \frac{\tan(k)+\tan(\ell)+\tan(k)\tan(\ell)}{1+\alpha\tan(k)\tan(\ell)}.
\end{align*}

The analogy with  trigonometric functions on $\mathbb R$ can be seen by viewing $\zeta$ as analogous to $e^i$ and $\bar\beta$ as analogous to $i$, with $\alpha$ analogous to $-1$; then the formulas for $\zeta^k$ and $\cos(k+\ell)$ become precise analogues of the familiar formulas.

\begin{rmk}
To complete the development of trigonometry over $\F_q$, we mention a new identity for odd $q$, using the notation from \cite{LS}:
\[
\tan_\zeta(x+y)=\frac{\tan_\zeta(x)+\tan_\zeta(y)}{1+i^2\tan_\zeta(x)\tan_\zeta(y)}.
\]
\end{rmk}

%#######################################################################
%#######################################################################

\section{Concluding Remarks}
In this note we showed that the tangent-Chebyshev rational functions introduced in \cite{LS} are conjugate to R\'edei functions, which explains the properties of these functions discovered in \cite{LS}, while also providing new properties.
Moreover, this enables one to deduce cryptographic information about tangent-Chebyshev maps from the known cryptographic results about R\'edei functions.
We also introduced a characteristic $2$ analogue of both tangent-Chebyshev rational functions and finite field trigonometry.
It seems plausible that one could similarly define tangent-Chebyshev rational maps on $\mathbb{Z}/n\mathbb{Z}$ by using the analogues of R\'edei functions on $\mathbb{Z}/n\mathbb{Z}$ introduced in \cite{NRup2,N2}, but we have not explored that possibility.

%#######################################################################
%#######################################################################
%#######################################################################

%\bibliographystyle{plain}

\end{document}